\documentclass[11pt]{article}

\usepackage{amssymb} 
\usepackage{amsmath}     
\usepackage{amsthm}
\usepackage{todonotes}
\usepackage{mathtools}
\usepackage{a4wide}
\usepackage{graphicx}
\usepackage{hyperref}

\usepackage[ruled,vlined,linesnumbered]{algorithm2e}
\SetKw{KwNot}{not}
\SetKw{KwOR}{or}
\SetKw{KwAND}{and}
\SetKw{KwExitFor}{exit for-loop}
\SetKw{KwStep}{Step}

\usepackage{authblk}

\newcommand{\Rset}{\mathbb{R}}

\newcommand{\Zset}{\mathbb{Z}}

\newtheorem{thm}{Theorem}
\newtheorem{lem}{Lemma}

\newtheorem{cor}{Corollary}

\newtheorem{rem}{Remark}

\newtheorem{cla}{Claim}

\begin{document}

\title{Computational complexity of the recoverable robust shortest path problem in acyclic digraphs under interval budgeted uncertainty}

\author[1]{Marcel Jackiewicz}
\author[1]{Adam Kasperski\footnote{Corresponding author}}
\author[1]{Pawe{\l} Zieli\'nski}

\affil[1]{
Wroc{\l}aw  University of Science and Technology, Wroc{\l}aw, Poland\\
            \texttt{\{marcel.jackiewicz,adam.kasperski,pawel.zielinski\}@pwr.edu.pl}}

\date{}

\maketitle

 \begin{abstract}
 In this paper, we consider the recoverable robust shortest path problem in acyclic digraphs, employing interval budgeted uncertainty to model uncertain second-stage costs. For the continuous budgeted uncertainty model, we prove that the problem is strongly NP-hard even in layered digraphs. Furthermore, we show that in general acyclic digraphs, the problem cannot be approximated within any constant factor unless $\mathrm{P} = \mathrm{NP}$, nor can it be approximated within a factor of $2^{\log^{1-\epsilon} n}$ for any $\epsilon > 0$ unless $\mathrm{NP} \subseteq \mathrm{DTIME}(n^{\mathrm{poly} \log n})$. For the discrete budgeted uncertainty model, we show that the problem  is not approximable unless $\mathrm{P} = \mathrm{NP}$, even in layered digraphs. Finally, we establish that under continuous budgeted uncertainty, the integrality gap of a relaxation allowing a fractional first-stage solution is at least $\Omega(\sqrt{n})$.
 
  \end{abstract}
  
\noindent \textbf{Keywords}: robust optimization, interval data, recovery, shortest path, computational complexity

\section{Problem formulation}

We are given a directed acyclic graph $G=(V,A)$. Two nodes in~$V$ are distinguished as a source $s\in V$ and a sink $t\in V$. Let $\Phi$ be the set of all $s$-$t$ paths in $G$ ($X\in \Phi$ is the set of arcs forming the path). For each arc $e\in A$, a first-stage cost $C_e\geq 0$ is specified. The second-stage arc costs are uncertain and we assume that they belong to an uncertainty set $\mathcal{U}\subseteq \Rset_{+}^{|A|}$. In the recoverable robust shortest path problem, we compute a first-stage path $X\in \Phi$. Then, after a second-stage cost scenario $\pmb{c}\in \mathcal{U}$ is revealed, the path $X$ can be modified to some extent. That is, we can choose a second-stage path $Y$ in some prescribed neighborhood $\Phi(X,k)$ of $X$, where $k\in \Zset_{+}$ is a given recovery parameter. We consider three types of neighborhoods, called arc-inclusion, arc-exclusion, and arc-symmetric difference, respectively~\cite{NO13}:
\begin{align}
\Phi^{\mathrm{incl}}(X,k)&=\{Y\in\Phi \,:\,  |Y\setminus X|\leq k\}, \label{incl}\\
\Phi^{\mathrm{excl}}(X,k)&=\{Y\in\Phi \,:\,  |X\setminus Y|\leq k\}, \label{excl}\\
\Phi^{\mathrm{sym}}(X,k)&=\{Y\in\Phi \,:\,  | (Y\setminus X) \cup (X\setminus Y) |\leq k\} \label{sym}.
\end{align}
The recoverable robust shortest path problem is defined as follows:
 \begin{equation}
\textsc{Rec Rob SP} : \; OPT=\min_{X \in \Phi} \left ( \sum_{e\in X} C_e  + 
\max_{\pmb{c}\in \mathcal{U}}\min_{Y\in \Phi(X,k)} \sum_{e\in Y} c_e \right).
\label{rrsp}
\end{equation}

In this paper, we use the interval uncertainty representation to model the uncertain second-stage arc costs. That is, for each arc $e\in A$ we specify a nominal second-stage cost $\hat{c}_e$ and the maximum deviation $\Delta_e$ of the second-stage cost from its nominal value. Hence, $c_e\in [\hat{c}_e, \hat{c}_e+\Delta_e]$, $e\in A$, and
\begin{align}
& \mathcal{U} =\{\pmb{c}\in \Rset_{+}^{|A|}\,:\, c_e\in [\hat{c}_e, \hat{c}_e+\Delta_e], e\in A\}.
\label{intset}
\end{align}
To control the amount of uncertainty, we consider two budgeted versions of $\mathcal{U}$, namely the discrete~\cite{BS04} and continuous~\cite{NO13} budgeted uncertainty sets, defined as follows:
\begin{align}
& \mathcal{U}(\Gamma^d)=\{\pmb{c}\in \mathcal{U} \,:\, |\{e\in A\,:\, c_e>\hat{c}_e\}|\leq \Gamma^d \},
\label{intsetgd}\\
& \mathcal{U}(\Gamma^c)=\{\pmb{c}\in \mathcal{U} \,:\, \sum_{e\in A} (c_e-\hat{c}_e) \leq \Gamma^c\}.
\label{intsetgc}
\end{align}The parameters $\Gamma^d$ and $\Gamma^c$ are called budgets. The discrete budget $\Gamma^d\in \{0,\dots,|A|\}$ denotes the maximum number of the second-stage costs that can differ from their nominal values. On the other hand, the continuous budget $\Gamma^c\geq 0$ denotes the maximum total deviation of the second-stage costs from their nominal values. Both uncertainty sets are commonly used in robust optimization as they allow us to reduce the price of robustness of the solutions computed (see, e.g.,~\cite{BS04}).

The concept of recoverable robustness was first proposed in~\cite{LLMS09}. 
The recoverable robust version of the shortest path problem was subsequently investigated in~\cite{B11, B12}. 
The problem is known to be strongly NP-hard for general digraphs under the interval uncertainty set $\mathcal{U}$. 
This hardness result remains true across all three types of neighborhoods~(\ref{incl})--(\ref{sym}) (see~\cite{B12, NO13, SAO09, JKZ24}). 
However, when the input graph $G$ is acyclic, the problem can be solved in polynomial time under the uncertainty set $\mathcal{U}$ for all neighborhoods~(\ref{incl})--(\ref{sym}), as recently shown in~\cite{JKZ24}. 
It turns out that adding budget constraints to $\mathcal{U}$ can significantly increase the computational complexity of the problem. 
Indeed, it has been shown in~\cite{GLW22} that the problem is weakly NP-hard under the uncertainty set $\mathcal{U}(\Gamma^d)$ even for a very restricted class of series-parallel acyclic digraphs. 
In contrast, the complexity of the problem under the uncertainty set $\mathcal{U}(\Gamma^c)$ has remained open until now.
In this paper, we provide several new complexity results for \textsc{Rec Rob SP} in acyclic digraphs. 
We show that the problem is strongly NP-hard for layered digraphs under both $\mathcal{U}(\Gamma^d)$ and $\mathcal{U}(\Gamma^c)$, and we establish some inapproximability lower bounds. 
All the complexity results are presented in Section~\ref{sec:compl}. 
Additionally, in Section~\ref{sec:relax}, we characterize the integrality gap of a relaxation of the problem under the uncertainty set $\mathcal{U}(\Gamma^c)$.

\section{Computational complexity results}
\label{sec:compl}

Consider first the \textsc{Rec Rob SP} problem for the continuous budgeted uncertainty set $\mathcal{U}(\Gamma^c)$. We prove the following result:
\begin{thm}\label{thm1}
    The \textsc{Rec Rob SP} problem is strongly NP-hard in layered digraphs with any neighborhood (\ref{incl})--(\ref{sym}),
    under $\mathcal{U}(\Gamma^c)$, if $k$ is part of the input and $\Gamma^c=1$.
\end{thm}
\begin{proof}
    We will show a reduction from the \textsc{Hamiltonian Path} problem, which is known to be strongly NP-complete (see, e.g.,~\cite{GJ79}).
    \begin{description}
        \item{\textsc{Hamiltonian Path}}
        \item{Input:} Given a directed graph $G=(V,A)$. 
        \item{Question:}  Does $G$ have a Hamiltonian path, that is, a simple directed path that visits all nodes of $G$ exactly once?
    \end{description}

    Let $G = (V, A)$, $|V| = n$, $|A| = m$, be the digraph of a given instance of \textsc{Hamiltonian Path}.
    The graph $G' = (V', A')$ in the corresponding instance of \textsc{Rec Rob SP}
    is built using the construction of Averbakh and Lebedev~\cite[Section 4]{AL04},
    who proposed it for a reduction from the Hamiltonian path problem to the min-max regret shortest path problem under interval uncertainty.
    The vertex set~$V'$ contains $2n$ copies of the vertices in $V$
    and two additional vertices $s$ and~$t$, i.e.,
    $V' = \{v_i^j : i \in [n], j \in [2n]\} \cup \{s, t\}$.
    The arc set~$A'$ is partitioned into two subsets: \emph{vertical} and \emph{diagonal} arcs.
    For each $i \in [n]$, the vertical arcs connect~$s$ to~$v_i^{1}$,
    $v_i^{j}$ to~$v_i^{j + 1}$ for $j \in [2n - 1]$,
    and $v_i^{2n}$ to $t$.
    These arcs create $n$ vertex-disjoint $s$-$t$ paths in~$G'$,
    and each of these paths has exactly~$2n + 1$ arcs.
    The diagonal arcs connect $v_i^{2j}$ to $v_{l}^{2j + 1}$ for $j \in [n - 1]$
    if and only if $(v_i, v_l) \in A$.
    The resulting digraph has a layered structure: the~$s$ and~$t$ vertices are in their own layers,
    and each remaining layer consists of the vertices~$\{v_i^{j} : i \in [n]\}$ for a fixed~$j$.
    It follows that there are $2n + 2$ layers in total.
    An example of this construction is shown in Figure~\ref{fig1}.

    \begin{figure}[htbp]        \centering
        \includegraphics[keepaspectratio,width=0.6\textwidth]{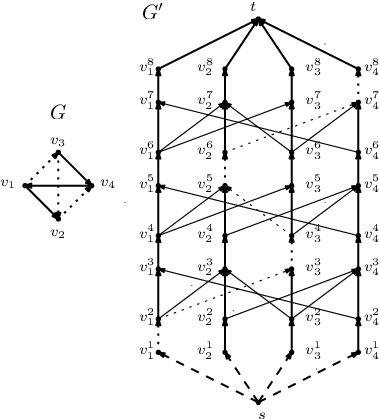}
        \caption{A sample digraph $G$ and the corresponding digraph $G'$.}
        \label{fig1}
    \end{figure}

    Having constructed~$G'$, we set the values of the \textsc{Rec Rob SP} parameters.
    The neighborhood-size parameter $k$ is set to~$2n$
    for the arc-inclusion and arc-exclusion neighborhoods, and~$4n$ for the arc-symmetric difference neighborhood.
    The uncertainty budget $\Gamma^c$ is set to~$1$.
    The first-stage cost of each arc in~$A'$ is set to~$0$.
    The second-stage cost depends on the type of the arc:
    \begin{equation}
        (\hat{c}_e, \Delta_e) =
        \begin{cases}
            (0, M) &\text{if } e \text{ is a vertical arc},\\
            (M, 0) &\text{if } e \text{ is a diagonal arc},
        \end{cases}
    \end{equation}
    where~$M$ is a constant such that $M > \Gamma^c$.
    Hence, the uncertainty affects only vertical arcs.

    We now prove that there exists a Hamiltonian path in~$G$ if and only if there exists a solution~$(X, Y)$
    to the constructed \textsc{Rec Rob SP} instance with total cost at most~$\Gamma^c/n$.

    ($\Rightarrow$) Let~$P$ be a Hamiltonian path in~$G$,
    and let the vertices traversed by that path be ordered as $v_{i_1}, v_{i_2}, \ldots, v_{i_n}$.
    We use $P$ to construct the first-stage path $X^*$ in~$G'$.
    Specifically, $X^*$ is composed of the arcs $(s, v_{i_1}^1),
    (v_{i_1}^1, v_{i_1}^2), (v_{i_1}^2, v_{i_2}^3), (v_{i_2}^3, v_{i_2}^4), \ldots,
    (v_{i_n}^{2n}, t)$,
    i.e., $X^*$ uses at least one arc from each vertex-disjoint vertical $s$-$t$ path $P_i$ in~$G'$ for~$i \in [n]$.
    Since $X^*$ shares at least one arc with each~$P_i$ and $|P_i| = 2n + 1$, we have $|P_i \setminus X^*| \leq 2n = k$,
    so $P_i \in \Phi(X^*,k)$ for every $i \in [n]$.
    Since the deviations affect only vertical arcs and the paths $P_i$ are arc-disjoint,
    the adversary maximizes the cost of the cheapest available recovery path
    by spreading the budget~$\Gamma^c$ evenly over the~$n$ paths: any uneven allocation leaves some~$P_i$ with cost below~$\Gamma^c/n$,
    which the recovery action would then select.
    Hence the worst-case second-stage cost is~$\Gamma^c/n$, and as all first-stage costs are zero,
    the total cost of the solution to \textsc{Rec Rob SP} is~$\Gamma^c/n$.

    ($\Leftarrow$) Let~$(X^*, Y^*)$ be a solution to \textsc{Rec Rob SP} in~$G'$ with a total cost of at most $\Gamma^c/n$.
    Since the first-stage cost of all arcs in~$G'$ is~$0$, and the second-stage cost of all diagonal arcs is~$M > \Gamma^c$,
    the second-stage path $Y^*$ must use only vertical arcs.
    It follows that $Y^*$ is one of the vertical vertex-disjoint $s$-$t$ paths~$P_i$ for~$i \in [n]$.
    Moreover, since $Y^* \in \Phi(X^*, k)$, the paths $X^*$ and~$Y^*$ must have at least one arc in common.
    We claim that $X^*$ shares at least one arc with each path $P_i$ for~$i \in [n]$.
    Suppose for contradiction that there is a path~$P_l$ for which~$|X^* \cap P_l| = 0$.
    Consequently, $P_l \notin \Phi(X^*, k)$, and the worst-case cost scenario is achieved by distributing the budget $\Gamma^c$
    across at most~$n - 1$ vertical paths.
    However, this results in a second-stage cost of at least~$\Gamma^c/(n - 1)$, which contradicts the assumption on the total solution cost.
    Therefore, $X^*$ contains at least one arc from each path~$P_i$ for $i \in [n]$,
    and it must use exactly~$n - 1$ diagonal arcs to form a valid $s$-$t$ path in~$G'$.
    By construction, a diagonal arc~$(v_u^{2j}, v_w^{2j + 1})$ belongs to~$A'$ if and only if an arc~$(v_u, v_w)$ belongs to~$A$.
    Hence, the sequence of diagonal arcs of~$X^*$ corresponds to a path $P$ in~$G$
    traversing each vertex in $V$ exactly once, i.e., a Hamiltonian path.
\end{proof}

The reduction of Theorem~\ref{thm1} proves the strong NP-hardness, but does not yield an inapproximability result.
In the following, we will show that the \textsc{Rec Rob SP} problem in general acyclic digraphs with any of the neighborhoods (\ref{incl})--(\ref{sym}),
under $\mathcal{U}(\Gamma^c)$ is not approximable within $2^{\log^{1-\epsilon} N}$ for any $\epsilon > 0$,
unless NP $\subseteq$ DTIME$(N^{\mathrm{poly} \log N})$, where $N$ is the input size,
and it is not approximable within any constant $\alpha > 1$ unless P=NP. 
In our reduction, we use the following Boolean satisfiability problems with identical inputs, which are known to be strongly NP-hard (see, e.g.,~\cite{GJ79}):
\begin{description}
    \item{Input:}  A 3-CNF formula of $n$ Boolean variables $\{x_i\}_{i\in[n]}$ given by $m$ clauses $\{\mathcal{C}_j\}_{j\in [m]}$.

    \item{\textsc{3-Sat}}
    \item{Question:} Is there a 0-1~assignment to the variables that satisfies all the clauses?

    \item{\textsc{Max 3-Sat}}
    \item{Output:} A 0-1 assignment to the variables that maximizes the number of satisfied clauses.
\end{description}
In the following, we will denote by $OPT(\Psi)$ the maximum value of the objective function in a \textsc{Max 3-Sat} instance~$\Psi$. We will use the following known result:

\begin{lem}[\cite{ALMSS92}]\label{lmaxsat}
    There is a fixed constant $\gamma>0$ and a polynomial time reduction~$\tau$ from \textsc{3-Sat} to \textsc{Max 3-Sat}
    such that for every \textsc{3-Sat} instance~$\Psi$:
    \begin{enumerate}
        \item if $\Psi$ is satisfiable then $OPT(\tau(\Psi))=m$,
        \item if $\Psi$ is not satisfiable then $OPT(\tau(\Psi)) < \dfrac{m}{1 + \gamma}$,
    \end{enumerate}
    where $m$ is the number of clauses in the \textsc{Max 3-Sat} instance~$\tau(\Psi)$.
\end{lem}
Lemma~\ref{lmaxsat} says, equivalently, that approximating \textsc{Max 3-Sat} within a constant factor of $1 + \gamma$ is NP-hard.
It was shown in~\cite{HA01} that such a factor is of~$8/7 - \epsilon$ for all $\epsilon > 0$, i.e.,
for all $\gamma < 1/7$.
Building on the gap reduction of Lemma~\ref{lmaxsat},
we present the gap-preserving reduction from \textsc{Max 3-Sat} to \textsc{Rec Rob SP} with acyclic output digraph.

\begin{lem}\label{lsatrrsp}
    Let $\Psi=(\{x_i\}_{i\in[n]}, \{\mathcal{C}_j\}_{j\in [m]})$ be an instance of \textsc{Max 3-Sat},
    $q$ be a positive integer such that $q < 2n$, and $\delta$ be any number from~$(0,1)$.
    Then there is a reduction
    from~$\Psi$ to an instance $\mathcal{I}^{(q)}=(G^{(q)}=(V^{(q)},A^{(q)}), (C^{(q)}_e)_{e\in A^{(q)}}, \mathcal{U}^{(q)}(\Gamma^c), \Phi(X,k))$ of \textsc{Rec Rob SP} for each of the neighborhoods (\ref{incl})-(\ref{sym}),
    where $G^{(q)}$ is an acyclic digraph with $|V^{(q)}| \leq n^{O(q)}$, $|A^{(q)}|\leq n^{O(q)}$,
    satisfying the following properties:
    \begin{enumerate}
        \item if $OPT(\Psi)=m$ then $OPT(\mathcal{I}^{(q)}) = \dfrac{1}{m^q}$,                  \label{satrrsp}
        \item if $OPT(\Psi)<\delta m$ then  $OPT(\mathcal{I}^{(q)}) > \dfrac{1}{\delta^q m^q}$. \label{rrspsat}
    \end{enumerate}
\end{lem}
\begin{proof}
    Let $\Psi=(\{x_i\}_{i\in[n]}, \{\mathcal{C}_j\}_{j\in [m]})$ be an instance of \textsc{Max 3-Sat}.
    W.l.o.g. we can assume that $x_i$ appears in some clause and its negation $\overline{x}_i$ appears in some clause.
    We construct an instance
    $\mathcal{I}^{(q)}=(G^{(q)}=(V^{(q)},A^{(q)}), (C^{(q)}_e)_{e\in A^{(q)}}, \mathcal{U}^{(q)}(\Gamma^c), \Phi(X,k))$
    of \textsc{Rec Rob SP} as follows. 
    The underlying acyclic digraph $G^{(q)}$ consists of $m^q$ \emph{clause gadgets}
    that correspond to the set of all $q$-tuples of the clause set $\mathfrak{C}=\{\mathcal{C}_1, \ldots, \mathcal{C}_m\}$.
    For every clause tuple $(\mathcal{C}^1, \ldots, \mathcal{C}^q) \in \mathfrak{C}^q$,
    where clause $\mathcal{C}^j \in \mathfrak{C}$, $j \in [q]$,
    has the form $\mathcal{C}^j = (l^{j}_{1} \vee l^{j}_2 \vee l^{j}_3)$,
    we build a clause gadget~$G^{(q)}_{(\mathcal{C}^{1}, \ldots, \mathcal{C}^{q})}$,
    which is a digraph that consists of~$2n$ horizontal layers
    corresponding to variables $x_i$ and their negations $\overline{x}_i$, $i \in [n]$,
    and has two distinguished nodes $s_{(\mathcal{C}^{1}, \ldots, \mathcal{C}^{q})}$ and~$t_{(\mathcal{C}^{1}, \ldots, \mathcal{C}^{q})}$.
    For each $q$-tuple of literals from $\mathcal{L}(\mathcal{C}^{1},\ldots, \mathcal{C}^{q}) = \{l_1^1, l_2^1, l_3^1\} \times \cdots \times \{l_1^q, l_2^q, l_3^q\}$,
    which do not contain contradictory literals, we create a \emph{literal path}
    from $s_{(\mathcal{C}^{1}, \ldots, \mathcal{C}^{q})}$ to~$t_{(\mathcal{C}^{1}, \ldots, \mathcal{C}^{q})}$.
    This path contains exactly $q$ literal arcs (solid fat arcs) that correspond to the $q$ components of the literal tuple.
    The rest of the arcs are dummy arcs (dotted arcs).
    The path created contains exactly $2n + 1$ arcs and contains exactly $q$ literal (solid fat) arcs and exactly $2n + 1 - q$ dummy arcs (dotted arcs).
    Examples are shown in Figures~\ref{gadget1} and~\ref{gadget2}.
    In Figure~\ref{gadget2} the leftmost path from $s_{(\mathcal{C}_1, \mathcal{C}_2)}$ to $t_{(\mathcal{C}_1, \mathcal{C}_2)}$
    corresponds to the tuple of literals $(x_1, x_1)$.
    Observe that two literal arcs are created in the layer corresponding to $x_1, \overline{x}_1$,
    and a path of length $2n + 1$ is then created by adding a number of dummy (dotted) arcs.
    There are at most $3^q$ vertical literal paths from $s_{(\mathcal{C}^{1}, \ldots, \mathcal{C}^{q})}$ to~$t_{(\mathcal{C}^{1}, \ldots, \mathcal{C}^{q})}$ in the gadget.
    They correspond to at most $3^q$ $q$-tuples of the literals
    from the literal set $\mathcal{L}(\mathcal{C}^{1}, \ldots, \mathcal{C}^{q})$.
    If the literal $x_i$ (resp. the literal $\overline{x}_i$), $i \in [n]$,
    is a component of some literal tuples in $\mathcal{L}(\mathcal{C}^{1},\ldots, \mathcal{C}^{q})$,
    then we connect the \emph{in-literal node} $x_i$ (resp. the in-literal node $\overline{x}_i$),
    the corresponding literal arcs and the out-literal node $x_i$ (resp. the out-literal node $\overline{x}_i$)
    using solid normal arcs to form a path from the in-literal node $x_i$ (resp. the in-literal node $\overline{x}_i$)
    to the out-literal node $x_i$ (resp. the out-literal node $\overline{x}_i$).
    Otherwise, we add a normal direct arc connecting the in-literal node $x_i$
    and the out-literal node $x_i$ or the in-literal node $\overline{x}_i$
    and the out-literal node $\overline{x}_i$.
    \begin{figure}
        \centering
        \includegraphics[keepaspectratio,width=0.9\textwidth]{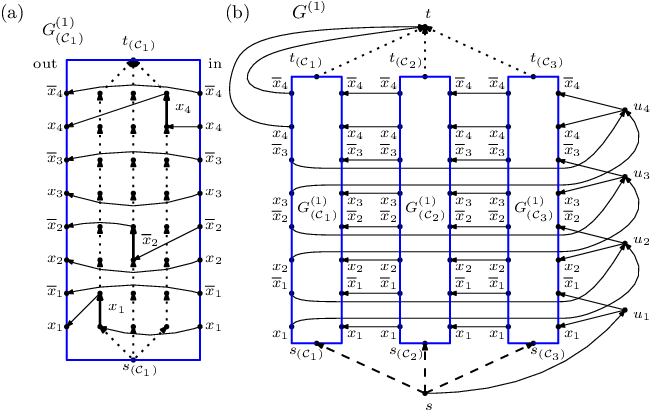}
        \caption{
            (a) A sample of clause gadget~$G^{(1)}_{(\mathcal{C}_1)}$  corresponding to the clause tuple $(\mathcal{C}_1)$,
            where $\mathcal{C}_1=(x_1 \vee \overline{x}_2 \vee x_4)$. 
            (b) A digraph $G^{(1)}$, that consists of three clause gadgets,
            corresponding to a formula~$\Psi=\mathcal{C}_1 \wedge \mathcal{C}_2 \wedge \mathcal{C}_3$.
        }
        \label{gadget1}
    \end{figure}
    \begin{figure}
        \centering
        \includegraphics[keepaspectratio,width=0.6\textwidth]{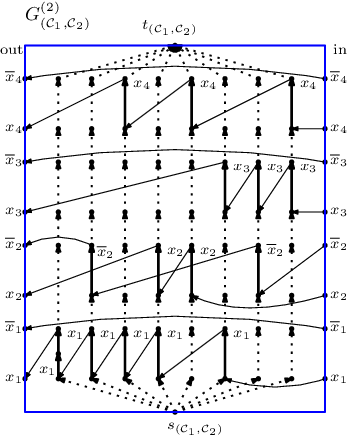}
        \caption{
            A sample clause gadget~$G^{(2)}_{(\mathcal{C}_1,\mathcal{C}_2)}$, corresponding to the clause tuple $(\mathcal{C}_1, \mathcal{C}_2)$,
            where $\mathcal{C}_1 = (x_1 \vee \overline{x}_2 \vee x_4)$ and $\mathcal{C}_2 = (x_1 \vee x_2 \vee x_3)$, includes 8 literal paths.
            The two leftmost literal paths correspond to the 2-tuples $(x_1, x_1)$ and $(\overline{x}_2, x_1)$, respectively.
            The literal path corresponding to $(\overline{x}_2, x_2)$ is omitted.
        }
        \label{gadget2}
    \end{figure}

    We now construct $G^{(q)} = (V^{(q)}, A^{(q)})$ as follows.
    We add to $G^{(q)}$ all the $m^q$ already built clause gadgets $G^{(q)}_{(\mathcal{C}^{1}, \ldots, \mathcal{C}^{q})}$,
    where $(\mathcal{C}^{1}, \ldots, \mathcal{C}^{q}) \in \mathfrak{C}^{q}$,
    $n$ \emph{variable nodes} $u_1, \ldots, u_n$ corresponding to variables $x_1, \ldots, x_n$, and two distinguished nodes $s$ and~$t$.
    For every $(\mathcal{C}^{1}, \ldots, \mathcal{C}^{q}) \in \mathfrak{C}^{q}$ we add two arcs:
    a dashed arc $(s, s_{(\mathcal{C}^{1}, \ldots, \mathcal{C}^{q})})$, a dotted arc $(t_{(\mathcal{C}^{1}, \ldots, \mathcal{C}^{q})}, t)$.
    Next, we add a normal arc $(s, u_1)$.
    Then, we form two disjoint paths from each variable node $u_i$ to the variable node $u_{i + 1}$, for $i \in [n - 1]$,
    which traverse through the~$m^q$ clause gadgets $G^{(q)}_{(\mathcal{C}^{1}, \ldots, \mathcal{C}^{q})}$,
    for $(\mathcal{C}^{1}, \ldots, \mathcal{C}^{q}) \in \mathfrak{C}^{q}$.
    This is achieved by adding normal arcs that go to the in-literal nodes $x_i$ and $\overline{x}_i$
    and emanating from the out-literal nodes $x_i$ and $\overline{x}_i$ of the gadgets.
    Similarly, for the variable node $u_n$,
    we form two disjoint paths going from $u_n$ to the node $t$.
    These paths are disjoint due to the construction of the clause gadgets.
    A sample construction of $G^{(1)} = (V^{(1)}, A^{(1)})$ is shown in Figure~\ref{gadget2}.

    The first- and the second-stage costs of an arc $e \in A^{(q)}$ are as follows:
    $C^{(q)}_e = 0$ if $e$ is solid (normal or fat),
    $C^{(q)}_e = 2$ if $e$ is dotted or dashed;
    $\hat{c}^{(q)}_e = 0$ and $\Delta^{(q)}_e = 0$ if $e$ is fat solid or dotted,
    $\hat{c}^{(q)}_e = 2$ and $\Delta^{(q)}_e = 0$ if $e$ is solid normal,
    $\hat{c}^{(q)}_e = 0$ and $\Delta^{(q)}_e = 1$ if $e$ is dashed.
    Notice that the uncertainty is only for the costs of the dashed arcs
    $(s, s_{(\mathcal{C}^{1}, \ldots, \mathcal{C}^{q})})$,
    $(\mathcal{C}^{1}, \ldots, \mathcal{C}^{q}) \in \mathfrak{C}^{q}$.
    Finally, we fix $\Gamma^c = 1$.
    We first consider the arc-inclusion neighborhood.
    We set $k = 2n + 3 - q$. Notice that $k > 0$, by the assumption that $q < 2n$.
    This completes the reduction. We now show that the reduction produces a gap of $1/\delta^q$.
    The following two claims result from the construction of the instance $\mathcal{I}^{(q)}$.

    \begin{cla}\label{cl11cor}
        There is a one-to-one correspondence between the set of assignments of the formula $\Psi$
        and the set of $s$-$t$ paths that use only solid arcs, both normal and fat.
    \end{cla}

    Indeed, consider any 0-1 assignment $\{x_i\}_{i \in [n]}$ to the variables.
    We form a path $X$ as follows: we add arc $(s, u_1)$ to $X$,
    and, for $i \in [n - 1]$ if $x_i = 1$ then we add the subpath from the variable node $u_i$
    to the variable node $u_{i+1}$ that traverses the in-literal nodes $x_i$ and the out-literal nodes $x_i$ of the clause gadgets;
    otherwise ($x_i = 0$) we add the subpath from the variable node $u_i$ to the variable node $u_{i + 1}$
    that traverses the in-literal nodes $\overline{x}_i$ and the out-literal nodes $\overline{x}_i$.
    Similarly, for the variable node $u_n$, we add one of the two disjoint subpaths from the variable node $u_{n}$ to $t$ depending on the value of $x_n$.
    Notice that $X$ is an $s$-$t$ path that uses only solid arcs (both normal and fat).
    Moreover, for two different 0-1 assignments, the $s$-$t$ paths formed in this way are distinct.
    It is easily seen that any path $X$ that uses only solid arcs (both normal and fat) determines a 0-1 assignment $\{x_i\}_{i \in [n]}$.
    Namely, if $X$ traverses the in-literal nodes $x_i$ and the out-literal nodes $x_i$ of the clause gadgets, then we set $x_i = 1$;
    otherwise, we set $x_i = 0$, which proves the claim.

    \begin{cla}\label{cloptx}
        Let $X$ be any $s$-$t$ path that uses only solid arcs, both normal and fat,
        and let $\ell$ be the number of clause gadgets that contain at least one path from $s$ to $t$
        that intersects exactly $q$ arcs with $X$ ($q$ literal arcs).
        If $\ell \geq 1$, then its total cost (the first- and second-stage costs) is $1/\ell$, where $1/m^{q} \leq 1/\ell \leq 1$;
        otherwise ($\ell = 0$), its total cost is at least $2$.
    \end{cla}

    The first-stage cost of $X$ is equal to $0$.
    Any path from $s$ to $t$ that goes through a literal path of a clause gadget intersecting exactly $q$ arcs with $X$
    is in the neighborhood $\Phi^{(\mathrm{incl})}(X, 2n + 3 - q)$,
    so we can replace $X$ with this path in the second stage.
    Therefore, in a worst second-stage scenario, the budget $\Gamma^c$ is shared uniformly among the dashed arcs,
    arcs $(s, s_{(\mathcal{C}^{1}, \ldots, \mathcal{C}^{q})})$,
    $(\mathcal{C}^{1}, \ldots, \mathcal{C}^{q}) \in \mathfrak{C}^q$, that contain these intersecting $s$-$t$ paths.
    Hence and from the facts that $\Gamma^c = 1$ and $|\mathfrak{C}^{q}| = m^q$,
    the total cost of $X$ is equal to $1/\ell$, where $1/m^q \leq 1/\ell \leq 1$.
    If $\ell = 0$, then no path from $s$ to $t$ through a clause gadget lies in the neighborhood $\Phi^{(\mathrm{incl})}(X, 2n + 3 - q)$:
    each such path has $2n + 3$ arcs and shares fewer than $q$ arcs with $X$, hence introduces more than $2n + 3 - q$ new arcs.
    Thus every recovery path in the neighborhood traverses a solid normal arc, whose second-stage cost is $2$, and so the total cost is at least $2$.

    It remains to prove the implications~\ref{satrrsp} and~\ref{rrspsat}.
    Suppose that there exists a 0-1 assignment~$\{x_i\}_{i \in [n]}$ to the variables that satisfies all $m$ clauses of the formula $\Psi$,
    i.e., $OPT(\Psi) = m$.
    By Claim~\ref{cl11cor}, the assignment~$\{x_i\}_{i \in [n]}$ uniquely determines the first-stage $s$-$t$ path~$X$,
    which uses only solid arcs (both normal and fat) and if $x_i = 1$, $i \in [n]$,
    then $X$ uses the arc $(u_i, x_i)$; otherwise $(x_i = 0)$ it uses the arc $(u_i, \overline{x}_i)$.
    Since the assignment satisfies $m$ clauses, every clause gadget
    $G^{(q)}_{(\mathcal{C}^{1},\ldots,\mathcal{C}^{q})}$, $(\mathcal{C}^{1},\ldots,\mathcal{C}^{q})\in \mathfrak{C}^q$,
    contains at least one path from $s$ to $t$, that intersects exactly $q$ literal arcs with $X$,
    corresponding to a literal tuple in $\mathcal{L}(\mathcal{C}^{1}, \ldots, \mathcal{C}^{q})$ satisfying the clauses $\mathcal{C}^{1}, \ldots, \mathcal{C}^{q}$.
    Thus, the total cost of $X$ is $1/m^q$, which follows from the fact that $|\mathfrak{C}^q| = m^q$.
    Again, by Claim~\ref{cloptx}, this is a lower bound on the cost of any optimal path.
    Hence, $X$ is an optimal path and $OPT(\mathcal{I}^{(q)}) = 1/m^q$.
    This completes the proof of the implication~\ref{satrrsp}.

    Assume that $OPT(\mathcal{I}^{(q)}) \leq 1/(\delta^q m^q)$.
    Let $X^*$ be an optimal $s$-$t$ path with the cost of $OPT(\mathcal{I}^{(q)})$.
    Clearly, it uses only solid arcs, both normal and fat, because other arcs have the first-stage cost equal to $2$.
    Let $\ell$ denote the number of clause gadgets that contain at least one path from $s$ to $t$ intersecting exactly $q$ literal arcs from $X^*$.
    These clause gadgets correspond to $\ell$ clause tuples in $\mathfrak{C}^q$.
    By Claim~\ref{cl11cor}, $X^*$ uniquely determines the 0-1 assignment $\{x_i^{*}\}_{i \in [n]}$ to the variables of $\Psi$.
    Therefore, in each of the $\ell$ clause gadgets, say, $G^{(q)}_{(\mathcal{C}^{1}, \ldots, \mathcal{C}^{q})}$,
    any $s$-$t$ path intersecting exactly $q$ literal arcs from $X^*$
    corresponds to a literal tuple in $\mathcal{L}(\mathcal{C}^{1}, \ldots, \mathcal{C}^{q})$
    that satisfies exactly the $q$ clauses $(\mathcal{C}^{1}, \ldots, \mathcal{C}^{q})$.
    From Claim~\ref{cloptx} and our assumption, it follows that the total cost of $X^*$ is such that $1/\ell \leq 1/(\delta^q m^q)$ and so $\ell \geq \delta^q m^q$.
    Now, a clause gadget $G^{(q)}_{(\mathcal{C}^{1}, \ldots, \mathcal{C}^{q})}$ is counted in $\ell$
    if and only if all $q$ of its clauses $\mathcal{C}^{1}, \ldots, \mathcal{C}^{q}$ are satisfied by the assignment $\mathbf{x}^*$.
    Since the gadgets correspond to all $q$-tuples in $\mathfrak{C}^{q}$,
    letting $S$ be the number of clauses of $\Psi$ satisfied by $\mathbf{x}^*$, exactly $\ell = S^{q}$ gadgets are counted.
    Hence $S^{q} = \ell \geq \delta^{q}m^{q}$, so that $S \geq \delta m$, i.e., $\mathbf{x}^*$ satisfies at least $\delta m$ clauses of the formula $\Psi$.
    From the optimality of $X^*$ and Claim~\ref{cl11cor}, we conclude that $OPT(\Psi) \geq \delta m$.
    This proves the implication~\ref{rrspsat}.
    Accordingly, the optimum value in the instance $\mathcal{I}^{(q)}$ of \textsc{Rec Rob SP}
    is either $1/m^q$ or greater than $1/(\delta m)^q$ and so the resulting gap is $1/\delta^q$.

    To modify the above reduction to obtain the same result for the arc-exclusion and arc-symmetric difference neighborhoods,
    we need to ensure that all optimal first-stage paths have the same number of arcs (normal or fat).
    Consider, for each $i \in [n - 1]$, the two subpaths from the variable node $u_i$ to the variable node $u_{i+1}$
    (resp. the two subpaths from $u_{n}$ to $t$) associated with $x_i = 1$ and $x_i = 0$
    (resp.\ $x_{n} = 1$ and $x_{n} = 0$), see Figure~\ref{gadget1}b.
    A raw upper bound on the number of arcs on any such subpath is $r = m^{q}(q\, 3^{q} + 1) + 1$.
    We split the first solid arc, $(u_i, x_i)$ or $(u_i, \overline{x}_i)$,
    on the subpath into a number of new solid arcs so that the number of arcs on the subpath equals $r$.
    Now, each first-stage path $X$ has exactly $n r + 1$ arcs.
    For the arc-exclusion neighborhood, we set $k = n r + 1 - q$,
    and for the arc-symmetric difference we fix $k = (2n + 3) + (n r + 1) - 2q$.
    The proofs of both implications for the arc-exclusion and arc-symmetric difference neighborhoods are the same as for the arc-inclusion neighborhood.

    It is easily seen that the instance $\mathcal{I}^{(q)}$ has size $N = n^{O(q)}$,
    where $n$ is the number of variables in the corresponding instance of \textsc{Max 3-Sat}.
\end{proof}

Combining Lemma~\ref{lmaxsat} (with $\delta = 1/(1 + \gamma)$)
and Lemma~\ref{lsatrrsp} gives a reduction from \textsc{3-Sat} to \textsc{Rec Rob SP}
whose satisfiable and unsatisfiable cases differ by a multiplicative gap exceeding $(1 + \gamma)^{q}$.
Choosing the amplifier $q$ as a slowly growing function of the input size turns this fixed gap
into an inapproximability factor that grows with the instance size,
at the price of a quasi-polynomial-time reduction, which yields the following bound.

\begin{thm}\label{tsatrrsp}
      The \textsc{Rec Rob SP} problem in acyclic digraphs under $\mathcal{U}(\Gamma^c)$, 
for each of the three neighborhoods (\ref{incl})--(\ref{sym}), 
cannot be approximated within a factor of $2^{\log^{1-\epsilon} N}$ for any $\epsilon > 0$, 
if $k$ is part of the input and $\Gamma^c=1$ (where $N$ is the input size), 
unless $\mathrm{NP} \subseteq \mathrm{DTIME}(N^{\operatorname{polylog} N})$.
\end{thm}
\begin{proof}
    Fix $\epsilon > 0$ and suppose that there exists a polynomial-time algorithm approximating
    \textsc{Rec Rob SP} in acyclic digraphs within a factor $2^{\log^{1-\epsilon} N}$, where $N$ is the input size.
    We show this implies $\mathrm{NP} \subseteq \mathrm{DTIME}(N^{\operatorname{polylog} N})$.

    Composing the reduction from \textsc{3-Sat} to \textsc{Max 3-Sat} (Lemma~\ref{lmaxsat})
    with the gap-preserving reduction of Lemma~\ref{lsatrrsp} yields,
    for $\delta = 1/(1 + \gamma) \in (0, 1)$, a reduction from \textsc{3-Sat} to \textsc{Rec Rob SP} such that
    \begin{enumerate}
        \item if $\phi$ is satisfiable then $OPT(\mathcal{I}^{(q)}) = \dfrac{1}{m^q}$,
        \item if $\phi$ is not satisfiable then $OPT(\mathcal{I}^{(q)}) > \dfrac{1}{\delta^{q} m^{q}},$
    \end{enumerate}
    so the two optimal values differ by a factor exceeding $(1/\delta)^{q}$,
    where $\phi$ is the Boolean formula of the \textsc{3-Sat} instance.

    Let $n$ and $m$ be the numbers of variables and clauses of the \textsc{Max 3-Sat} instance $\Psi = \tau(\phi)$.
    Fix $\beta > 1/\epsilon - 1$ and set $q = \lceil \log^{\beta}n \rceil$,
    which satisfies $q < 2n$ for all sufficiently large $n$, so the lemma applies.
    By Lemma~\ref{lsatrrsp} the reduction runs in time $n^{O(q)} = n^{O(\log^{\beta}n)} = 2^{O(\log^{\beta + 1}n)}$,
    which is quasi-polynomial, and the produced instance has size $N = n^{O(q)}$,
    so $\log N = O(\log^{\beta + 1} n)$.
    Consequently $\log^{1 - \epsilon} N = O(\log^{(\beta + 1)(1 - \epsilon)} n)$,
    whereas $\log{((1/\delta)^{q})} = q \log(1/\delta) = \Theta(\log^{\beta}n)$.
    Since $\beta > 1/\epsilon - 1$, we have $(\beta + 1)(1 - \epsilon) < \beta$,
    hence $\log^{1 - \epsilon}N = O(\log^{(\beta + 1)(1 - \epsilon)}n) = o(\log^{\beta}n) = o(\log((1/\delta)^{q}))$,
    so $2^{\log^{1 - \epsilon}N} < (1/\delta)^{q}$ for all sufficiently large~$n$.

Assume that there exists an algorithm for the \textsc{Rec Rob SP} problem
that achieves an approximation ratio of $2^{\log^{1-\epsilon} N}$ and runs in polynomial time $O(N^c)$ for some constant $c \ge 1$.
Applying this algorithm to the constructed instance, on a satisfiable $\phi$ it returns a value of at most
$2^{\log^{1-\epsilon} N}/m^{q} < (1/\delta)^{q}/m^{q}$,
and on an unsatisfiable one a value strictly above $(1/\delta)^{q}/m^{q}$.
Comparing the returned value with this threshold therefore correctly decides whether $\phi$ is satisfiable.
Since the reduction takes $n^{O(\log^{\beta}n)}$ time and the approximation algorithm runs in $O(N^c)$ time,
the total time to decide an arbitrary \textsc{3-Sat} instance is $(n^{O(\log^{\beta}n)})^c = n^{O(\log^{\beta}n)}$, 
which is bounded by $N^{\operatorname{polylog} N}$.
Since \textsc{3-Sat} is NP-complete, this would imply $\mathrm{NP} \subseteq \mathrm{DTIME}(N^{\operatorname{polylog} N})$.
\end{proof}

Theorem~\ref{tsatrrsp} rules out even sub-polynomial approximation factors,
but only under the assumption $\mathrm{NP} \not\subseteq \mathrm{DTIME}(N^{\operatorname{polylog} N})$.
Taking the amplifier to be constant trades this for a polynomial-time reduction with a constant gap,
ruling out approximation within any constant factor under the standard assumption $\mathrm{P} \neq \mathrm{NP}$.

\begin{cor}\label{corhard}    
    For each of the three neighborhoods (\ref{incl})--(\ref{sym}), 
the \textsc{Rec Rob SP} problem in acyclic digraphs under $\mathcal{U}(\Gamma^c)$ 
admits no polynomial-time $\alpha$-approximation for any constant $\alpha > 1$, if $k$ is part of the input and $\Gamma^c=1$,
unless $\mathrm{P} = \mathrm{NP}$.
\end{cor}
\begin{proof}
    Fix a constant $\alpha > 1$ and apply the reduction underlying Theorem~\ref{tsatrrsp}
    with the constant gap amplifier $q= \lceil\log\alpha/\log(1/\delta)\rceil$,
    so that $(1/\delta)^{q} \geq \alpha$ with $\delta = 1/(1 + \gamma)$.
    Let $n$ be the number of variables of the instance $\Psi$ of \textsc{Max 3-Sat} on which Lemma~\ref{lsatrrsp} is invoked.
    As $q$ is a fixed constant, the condition $q < 2n$ holds for every instance with $n > q/2$;
    the remaining instances, with $n \leq q/2$, have $O(1)$ variables (at most $q/2$, a constant)
    and are decided directly by exhaustive search over their $O(1)$ assignments.
    We may therefore assume $q < 2n$.
    Since $q$ is constant, the reduction runs in time $n^{O(q)} = n^{O(1)}$, i.e., in polynomial time,
    and produces an acyclic instance $\mathcal{I}^{(q)}$ with $OPT(\mathcal{I}^{(q)}) = 1/m^q$
    when $\phi$ (the formula in the initial \textsc{3-Sat} instance) is satisfiable
    and $OPT(\mathcal{I}^{(q)}) > \alpha/m^q$ otherwise.
    A polynomial-time $\alpha$-approximation algorithm would return a value at most $\alpha/m^q$ in the former case
    and more than $\alpha/m^q$ in the latter, deciding the satisfiability of $\phi$ in polynomial time,
    which would imply $\mathrm{P} = \mathrm{NP}$.
\end{proof}


We now consider discrete budgeted uncertainty set $\mathcal{U}(\Gamma^d)$.
Recall that \textsc{Rec Rob SP} is then weakly NP-hard for a very restrictive class digraphs:
\begin{thm}[\cite{GLW22}]
The \textsc{Rec Rob SP} problem is weakly NP-hard for arc series-parallel digraphs with any neighborhood (\ref{incl})-(\ref{sym}),
under $\mathcal{U}(\Gamma^d)$,
 if $k$ is part of the input.  
\end{thm}
In the case of more general acyclic digraphs, the problem becomes significantly more challenging. We prove the following stronger result if the input graph is layered:
\begin{thm}
The \textsc{Rec Rob SP} problem is strongly NP-hard for layered digraphs with  any neighborhood (\ref{incl})-(\ref{sym}),
under $\mathcal{U}(\Gamma^d)$, and it is also not at all approximable, unless P=NP,
 if $k$ is part of the input.  
\end{thm}
\begin{proof}
Given an instance $G$ of \textsc{Hamiltonian Path}, the construction of $G'$ is the same as in the proof of Theorem~\ref{thm1}. We again set the recovery parameter $k=2n$ for the arc-inclusion (arc-exclusion) neighborhood and $k=4n$ for the arc-symmetric difference neighborhood. The first-stage costs of all arcs in $G'$  are~0. The second-stage costs of all vertical arcs in $G'$ are in the interval $[0,0]$, except for the \emph{dashed} arcs $(s,v_i^1)$, $i\in [n]$, whose second-stage costs are in the interval $[0,1]$ (see the dashed arcs in Figure~\ref{fig1}). The second-stage costs of all diagonal arcs are in the interval $[1,1]$.
 Let $\Gamma^d=n-1$. We claim that there is a Hamiltonian path in $G$ if and only if there is a solution to 
 \textsc{Rec Rob SP} in~$G'$ with the cost equal to~0.
 The proof of this claim is based on the same reasoning as in the proof of Theorem~\ref{thm1}. Graph $G$ has a Hamiltonian path if and only if there is a first-stage path $X$ in $G'$ that uses at least one arc from each vertical path $P_1,\dots, P_n$.
Because, in the worst second-stage cost scenario, we can set the cost of at most $n-1$ dashed arcs to~1, $G$ has a Hamiltonian path if and only if there is a second-stage path in $\Phi(X,k)$ with the cost equal to~0. Using the fact that the first-stage cost of each solution equals~0, we are done.
\end{proof}

It is worth pointing out that, in contrast to the uncertainty set $\mathcal{U}(\Gamma^d)$, it is possible to provide an approximation guarantee for the problem under $\mathcal{U}(\Gamma^c)$. 
Let $D=\sum_{e\in A} \Delta_e$ and define a second-stage cost scenario $\pmb{c}'\in \mathcal{U}(\Gamma^c)$ where $c'_e=\min\left\{\hat{c}_e+\Delta_e,\hat{c}_e+\Gamma^c\frac{\Delta_e}{D}\right\}$ for $e\in A$. 
We can solve the \textsc{Rec Rob SP} problem in which $\mathcal{U}=\{\pmb{c}'\}$ (i.e., there is only one second-stage cost scenario) in polynomial time using the algorithm proposed in~\cite{JKZ24}. 
By the same reasoning as in~\cite[Lemma~7]{HKZ16a}, we can show that the total cost of an optimal solution to this problem is 
at most ${|A|} \cdot \mathrm{OPT}$, where $\mathrm{OPT}$ is the total cost of an optimal solution to the problem under $\mathcal{U}(\Gamma^c)$.

\begin{rem}
The \textsc{Rec Rob SP} problem in acyclic digraphs under $\mathcal{U}(\Gamma^c)$ is approximable within $|A|$.
\end{rem}

\section{Integrality gap of a relaxation of \textsc{Rec Rob SP} under~$\mathcal{U}(\Gamma^c)$}
\label{sec:relax}
 In this section we investigate a relaxation of \textsc{Rec Rob SP} under~$\mathcal{U}(\Gamma^c)$, in which a fractional solution can be chosen in the first stage.
  Let $\chi(\Phi)\subseteq \{0,1\}^{|A|}$ be the set of characteristic vectors of the $s$-$t$ paths in $G=(V,A)$, defined by the following standard flow conservation constraints:
  $$\chi(\Phi)=\left\{\pmb{x}\in\{0,1\}^{|A|}: \sum_{e\in {\rm Out}(i)} x_e-\sum_{e\in {\rm In}(i)} x_e=d_i,\; i\in V \right\},$$
 where $d_i=1$ if $i=s$, $d_i=-1$ if $i=t$ and $d_i=0$ if $i\in V\setminus\{s,t\}$; ${\rm Out}(i)$, ${\rm In}(i)$ are the sets of outcoming and incoming arcs for  $i\in V$, respectively.  Since $G$ is acyclic, only simple $s$-$t$ paths are included in $\chi(\Phi)$.
The set of characteristic vectors of the second-stage paths, for instance, for the arc-inclusion neighborhood can be defined as follows:
$$\chi(\Phi^{\rm incl}(\pmb{x},k))=\left\{\pmb{y}\in \chi(\Phi): \sum_{e\in A} (1-x_e)y_{e}\leq k\right\}.$$
We can now represent \textsc{Rec Rob SP} under the uncertainty set $\mathcal{U}(\Gamma^c)$  as
\begin{equation}
\label{recrobbin}
\min_{\pmb{x}\in \chi(\Phi)} \left(\pmb{C}^T\pmb{x}+\max_{\pmb{c}\in \mathcal{U}(\Gamma^c)} \min_{\pmb{y}\in  \chi(\Phi(\pmb{x},k))} \pmb{c}^T\pmb{y}\right).
\end{equation}
 Let $\chi'(\Phi)$ be a relaxed version of the set $\chi(\Phi)$ in which $\pmb{x}\in\{0,1\}^{|A|}$ is replaced with $\pmb{x}\in [0,1]^{|A|}$. Consider the following relaxation of~(\ref{recrobbin}):
 \begin{equation}
\label{recrobbinrel}
\min_{\pmb{x}\in \chi'(\Phi)} \left(\pmb{C}^T\pmb{x}+\max_{\pmb{c}\in \mathcal{U}(\Gamma^c)} \min_{\pmb{y}\in  \chi(\Phi(\pmb{x},k))} \pmb{c}^T\pmb{y}\right).
\end{equation}
Note that we can only choose a relaxed solution in the first stage.
The next theorem characterizes the integrality gap of~(\ref{recrobbinrel}).

\begin{thm}
\label{thmgap}
The integrality gap of~(\ref{recrobbinrel}) is at least $\Omega(\sqrt{{|V|}})$ for each of the three neighborhoods (\ref{incl})--(\ref{sym}).
\end{thm}
\begin{proof}
Consider a digraph shown in Figure~\ref{figrelax}.
\begin{figure}[h!t]
    \centering
    \includegraphics{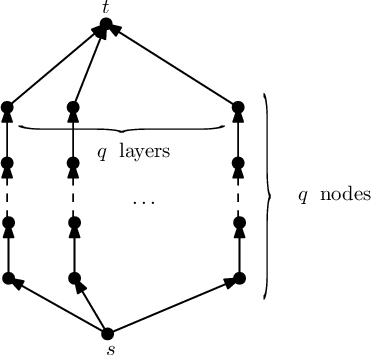}
    \caption{
        A graph $G=(V,A)$, $C_e=0$, $\hat{c}_e=0$, $\Delta_e=1$ for each $e\in A$, $\Gamma^c=q$, $k=q$ for the arc-inclusion and arc-exclusion neighborhoods,
        and $k=2q$ for the arc-symmetric difference neighborhood.}
    \label{figrelax}
\end{figure}
This graph has  the source~$s$, the sink~$t$, and $q$~layers, each of which containing $q$ nodes. Therefore, $G$ has $q^2+2$ nodes and $q(q-1)+2q$ arcs. We set $C_e=0$, $\hat{c}_e=0$, $\Delta_e=1$ for each $e\in A$, and $\Gamma^c=q$. 
The recovery parameter is set to $k=q$ for the arc-inclusion and arc-exclusion neighborhoods, and to $k=2q$ for the arc-symmetric difference neighborhood.
We prove the result for the arc-inclusion neighborhood. However, because $G$ is layered, the same reasoning applies  to all the neighborhoods. 
Graph $G$ is composed of exactly $q$ disjoint paths and each such a path has the length equal to~$q+1$.
An optimal solution to \textsc{Rec Rob SP} is any $s$-$t$ path in $G$. 
Since this path cannot be modified in the second stage, its overall cost is $q=\Omega(\sqrt{{|V|}})$, because the whole budget $\Gamma^c=q$ can be allocated to this selected path.
Consider now a fractional solution $\pmb{x}'\geq \pmb{0}$ in which  $x'_e=\frac{1}{q}$ for each $e\in A$. It is easy to see that $\pmb{x}'\in \chi'(\Phi)$.
 Now each path $\pmb{y} \in \chi(\Phi)$ is in $\chi(\Phi^{\rm incl}(\pmb{x}',k))$. Indeed,
 $$\sum_{e\in A} (1-x'_e)y_e=(q+1)\left(1-\frac{1}{q}\right)=q-\frac{1}{q}\leq q$$
 and we have fixed $k=q$.
Therefore, the adversary must split the budget $\Gamma^c=q$ equally among all $q$ paths in $G$. In what follows, the total cost of $\pmb{x}'$ is~1. So, the integrality gap of~(\ref{recrobbinrel}) is at least $\Omega(\sqrt{{|V|}})$.
\end{proof}

\section{Conclusions}

In this paper, we have provided new computational complexity results for the recoverable robust shortest path problem in acyclic digraphs. We used two types of interval budgeted uncertainty sets to model uncertain second-stage costs. The problem turned out to be strongly NP-hard and inapproximable under discrete budgeted uncertainty. These negative results hold even for layered digraphs. An interesting open problem is to strengthen the known hardness result for series-parallel digraphs, for which the problem is currently only known to be weakly NP-hard. The complexity landscape under continuous budgeted uncertainty is considerably more intriguing. For general acyclic digraphs, the problem is approximable within a factor of $|A|$, but hard to approximate within $2^{\log^{1-\epsilon} N}$ for any $\epsilon>0$. Thus, there is still room to improve both the lower and upper approximation bounds. Furthermore, the complexity of the problem for series-parallel digraphs remains open.

\subsubsection*{Acknowledgements}
The authors were supported by
 the National Science Centre, Poland, grant 2022/45/B/HS4/00355.


\end{document}